\documentclass{amc}
\usepackage{paralist}
 \usepackage{epsfig} 

  \def\currentyear{2012}
   
    \def\ppages{}
     \def\DOI{}

\newtheorem{theorem}{Theorem}[section]
\newtheorem{corollary}[theorem]{Corollary}

\newtheorem{lemma}[theorem]{Lemma}
\newtheorem{proposition}[theorem]{Proposition}

\theoremstyle{definition}
\newtheorem{ex}[theorem]{Example}
\newtheorem{definition}[theorem]{Definition}

\newcommand{\tr}{{\operatorname{Tr}}}

\title[Interleavers from permutation functions]
      {Cycle structure of permutation functions over finite fields and their applications}

\author[Amin Sakzad, Mohammad-Reza Sadeghi and Daniel Panario]{}


\thanks{Part of the material in this paper was presented at 48th Annual Allerton Conference, USA, 2010.}

\begin{document}
\maketitle

\centerline{\scshape Amin Sakzad and Mohammad-Reza Sadeghi}
\medskip
{\footnotesize
 \centerline{Department of Mathematics and Computer Science}
  \centerline{Amirkabir University of Technology, Tehran, Iran}
}

\bigskip

\centerline{\scshape Daniel Panario}
\medskip
{\footnotesize
 \centerline{School of Mathematics and Statistics}
  \centerline{Carleton University, Ottawa, Canada}
}

\bigskip


\begin{abstract}
In this work we establish some new interleavers based on permutation functions.
The inverses of these interleavers are known over a finite field $\mathbb{F}_q$.
For the first time M\"{o}bius and R\'edei functions are used to give new
deterministic interleavers. Furthermore we employ Skolem sequences
in order to find new interleavers with known cycle structure.
In the case of R\'edei functions an exact formula for the inverse function is derived.
The cycle structure of R\'edei functions is also investigated.
The self-inverse and non-self-inverse versions of these permutation
functions can be used to construct new interleavers.
\end{abstract}

\section{Introduction}
\label{Introduction}
Interleavers have a lot of impact on various aspects of Communication Theory.
For example, a crucial role in designing good turbo codes and turbo lattices
is played by interleavers~\cite{costello, sakzad10, Sakzad10}.
The distance properties of turbo codes can change dramatically from one interleaver to another~\cite{Boutros06}.
In addition, bit interleavers are employed at the encoder of a coded modulation scheme
over fading channels to improve the error performance~\cite{Caire98}.

Several studies have been conducted on deterministic interleavers.
One of the main advantages of well-known deterministic
interleavers is that they have a simple
structure that is easy to implement. Only some defining parameters of
the interleaver, such as the coefficients of a polynomial, are stored.
Efforts in the field of deterministic
interleavers have focused on permutation polynomials~\cite{sun,tak1}.
Therefore, it seems natural to search in
the class of permutation functions for finding good interleavers.
Recent works concerning the inverse of
interleavers like~\cite{cheng,moision,ryu,takcost} motivated us to
turn our attention to permutations which their inverses
are given. Interleavers with known inverses
are of interest~\cite{cheng,moision} because the same
structure and technology used for encoding can be used for decoding as well.
The permutation polynomials used in these deterministic interleavers are all over the
integer ring $\mathbb{Z}_n$.

Providing interleavers based on permutation functions is the main
contribution of this work.
In this paper we use permutation functions over $\mathbb{F}_q$ and
Skolem sequences to create interleavers with known inverses.
Using results from finite fields, we can construct permutation
polynomials and permutation functions over $\mathbb{F}_q$.
However, only very special cases of permutation functions are
known. Some examples are permutation monomials, Dickson permutation
polynomials, nonlinear transformation (M\"{o}bius) and R\'edei
permutation functions.

The study of permutation monomials $x^n$ with a cycle of length
$j$ has been treated in~\cite{ahmad}. Permutation monomials $x^n$
with all cycles of the same length are characterized in~\cite{rob2}.
The cycle structure of Dickson permutation polynomials $D_n(x,a)$
where $a\in\{0,\pm1\}$ has been studied in~\cite{mullen}. The
cycle structure of M\"{o}bius transformation has been described
in~\cite{turkey}. In this paper, we provide an exact formula for the
inverse of every R\'edei function and study the cycle structure
of R\'edei functions. More precisely, R\'edei functions
with a cycle of length $j$ are characterized. Then we extend this
to all cycles of the same length $j$ or $1$. An exact formula
for the number of cycles of length $j$ is given. These
are other contributions of this study.

Skolem sequences have been introduced and
studied extensively~\cite{cathybaker,surveySS}. These sequences
have applications in constructing cyclic Steiner triple systems
and in constructing codes resistant to random
interference~\cite{handbook,skolemapplication}. In this work we
continue to find applications for these nice structured sequences.
Specially, we use these sequences to produce self-inverse
and non-self-inverse interleavers.

This paper is organized as follows: For making the paper self-contained,
background on interleavers and on permutation functions are given in
Section~\ref{BasicDefinitionsandBackground}. The general structure of
our deterministic interleavers is explained and investigated in
this section.
Monomial, Dickson, M\"{o}bius, R\'edei and Skolem interleavers are
studied in Section~\ref{NewInterleaversFq}.
The cycle structure of R\'edei functions as well as the number of
cycles of certain length are also given in this section.
Conclusions and further work are commented in Section~\ref{conclusion}.

\section{Background}~\label{BasicDefinitionsandBackground}

\subsection{Basic definition of interleavers}~\label{interleavers}

Let us first give the general concept of an interleaver.
An interleaver
$\Pi$ may be interpreted as a function which permutes the indices
of components of ${\bf u}$. In other words, let
$I=\{0,1,\ldots, N-1\}$ be all indices of a vector ${\bf u}=(u_0,\ldots,u_{N-1})$, then the
interleaver $\Pi$ can be considered as a bijective function
of $I$. The inverse function $\Pi^{-1}$ is also necessary for
decoding process when we implement a \emph{deinterleaver}. An
interleaver $\Pi$ is called \emph{self-inverse} if $\Pi=\Pi^{-1}$.
\medskip

\subsection{Some well-known permutation functions over finite fields}~\label{PP}

Let $q=p^m$ and $\mathbb{F}_q$ be the
finite field of order $q$ where $p$ is a prime number.
A \emph{permutation function} over
$\mathbb{F}_q$ is a bijection that sends the elements
of $\mathbb{F}_q$ onto itself. It is clear that permutation functions
have a functional inverse with respect to composition. Thus,
for a permutation function $P\in\mathbb{F}_q[x]$, there exists
a unique $P^{-1}\in\mathbb{F}_q[x]$ of degree less than $q$
such that $P(P^{-1}(x))=P^{-1}(P(x))=x\pmod{x^q-x}$ for all
$x\in\mathbb{F}_q$. A permutation function $P$ is called
\emph{self-inverse} if $P=P^{-1}$.

Let $f$ be a primitive polynomial of degree $m$ over
$\mathbb{F}_p$ and assume that $\alpha$ is a root of $f$.
Since $q=p^m$ and $f$ divides $x^{q-1}-1$, we can represent
$\mathbb{F}_q$ as
\begin{equation}~\label{FF}
\mathbb{F}_q=\{0,\alpha^1,\ldots,\alpha^{q-2},\alpha^{q-1}\},
\end{equation}
where $\alpha^{q-1}=1$. The above representation (power
representation) of $\mathbb{F}_q$ is appropriate for
operations like multiplication and raising to a power.
Furthermore, for every $\alpha^i$, $0\leq i\leq q-1$,
there exists a polynomial representation (in $\alpha$)
with degree less than $m$ which is adequate for addition
and subtraction.

Next, we review four well-known permutation functions on
the finite field $\mathbb{F}_{q}$. They are useful for
constructing new deterministic interleavers.
\begin{itemize}
\item{Monomials~\cite{lidl}:} $M(x)=x^n$ for some $n\in \mathbb{N}$
is a permutation polynomial over $\mathbb{F}_q$ if and only if $\gcd(n,q-1)=1$ where $\gcd$ denotes the greatest common
divisor. The inverse of $M(x)$ is obviously the monomial
$M^{-1}(x)=x^m$ where $nm\equiv1\pmod{q-1}$.
\item{Dickson polynomials of the 1st kind~\cite{lidl}:}
Let $n$ be an integer. A function $D_n$ which satisfies
$D_n(x+y,xy)=x^n+y^n$ for all $x,y\in\mathbb{F}_q$
is called a Dickson
polynomial over $\mathbb{F}_q$. Let us fix $y=a$ for an element
$a\in\mathbb{F}_q$. The function $D_n(x,a)$ is a permutation polynomial
if and only if $\gcd(n,q^2-1)=1$.
For $a\in\{0,\pm 1\}$, the inverse of $D_n(x,a)$ is
$D_m(x,a)$ where $nm\equiv1\pmod{q^2-1}$.
It is easy to check~\cite{lidl} that
\begin{equation}~\label{DickPP}
    D_n(x,a)=\sum_{i=0}^{\lfloor n/2\rfloor}
        \frac{n}{n-i} \binom{n-i}{i}(-a)^px^{n-2i}.
\end{equation}
\item{M\"{o}bius transformation:} The function
\begin{equation}~\label{MobPP}
T(x)=\left\{
\begin{array}{ll}
 \frac{ax+b}{cx+d}& x\neq \frac{-d}{c},\\
\frac{a}{c}& x=\frac{-d}{c},
\end{array}
\right.
\end{equation}
where $a,b,c,d\in\mathbb{F}_q$, $c\neq0$ and $ad-bc\neq0$ is a
permutation function. Its inverse is simply
\begin{equation}~\label{MobinversePP}
T^{-1}(x)=\left\{
\begin{array}{ll}
 \frac{dx-b}{-cx+a}& x\neq \frac{a}{c},\\
\frac{-d}{c}& x=\frac{a}{c}.
\end{array}
\right.
\end{equation}
\item{R\'edei functions~\cite{Redei}:}
Let $\mbox{char}(\mathbb{F}_q)\neq2$ and $a\in \mathbb{F}_{q}^\ast$
be a non-square element; then $\sqrt{a}\in\mathbb{F}_{q^2}^\ast$.
The numerator $G_n(x,a)$ and the denominator $H_n(x,a)$
of the R\'edei function are
polynomials in $\mathbb{F}_q[x]$ satisfying the equation
$$(x+\sqrt{a})^n=G_n(x,a)+H_n(x,a)\sqrt{a}.$$
Also, it is easy to see that
$$(x-\sqrt{a})^n=G_n(x,a)-H_n(x,a)\sqrt{a}.$$
The R\'edei function $\displaystyle R_n=\frac{G_n}{H_n}$ with degree $n$
is a rational function over $\mathbb{F}_q$. We have that $R_n$
is a permutation function if and only if $\gcd(n,q+1)=1$.
In addition, if
$\mbox{char}(\mathbb{F}_q)\neq2$ and $a\in \mathbb{F}_{q}^\ast$
is a square element, then $R_n$ is a permutation function if and only if
$(n,q-1)=1$.

\end{itemize}
\medskip
\subsection{Skolem sequences}~\label{SS}

Let $D$ be a set of integers. A \emph{Skolem-type sequence} is a sequence with alphabet $D$
where each element $i\in D$ appears exactly twice in the sequence at positions
$a_i$ and $a_i+i=b_i$. Thus, $|b_i-a_i|=i$ for every $i\in D$. These sequences might have empty
positions, which we fill with zeros. For more information about Skolem sequences, we refer
the reader to~\cite{surveySS}. Here we use Skolem sequences to
construct self-inverse interleavers of a specific size.
A partition of the set $[n]=\{1,\ldots,n\}$ into $n$ ordered pairs
$\{(a_i,b_i)\colon b_i-a_i=i,~1\leq i\leq n\}$,
implies a \emph{Skolem sequence of order $n$}. It is obvious that in order to
generate Skolem sequence corresponding to this partition we must put integer
$i\in[n]$ in positions $a_i$ and $b_i$ of a sequence $S=(s_1,\ldots,s_{2n})$.
A \emph{$k$-extended Skolem sequence}  of order $n$ is a Skolem sequence of order $n$
which contains exactly one hole in position $k$. If $k$ is in the penultimate position,
the sequence is called a \emph{hooked sequence}.
A \emph{$(j,n)$-generalized Skolem sequence}  of multiplicity $j$
is a sequence $S=(s_1,\ldots,s_t)$ of integers from $[n]$
such that for every $i\in[n]$ there are exactly $j$ positions in the sequence $S$,
let us say $r_1,r_2=r_1+i,\ldots,r_j=r_1+(j-1)i$, such that $s_{r_1}=s_{r_2}=\cdots=s_{r_j}=i$.
It is easy to see that $t=jn$ in this case.

\begin{ex}
The sequence $(4,1,1,3,4,2,3,2)$ is a Skolem sequence of order $4$. The sequence
$(2,5,2,6,1,1,5,3,4,6,3,0,4)$
is a hooked Skolem sequence of order $6$.
\end{ex}

In the following we mention two theorems from~\cite{surveySS}
which state necessary and sufficient conditions for
the existence of the above defined Skolem sequences.

\begin{theorem}\label{th:SSkextendedSShookedSS}
A Skolem sequence of order $n$ exists if and only if $n\equiv0,1\pmod{4}$.
A hooked Skolem sequence of order $n$ exists if and only if $n\equiv2,3\pmod{4}$.
A $k$-extended Skolem sequence of order $n$ exists if and only if $n\equiv0,1\pmod{4}$,
when $k$ is odd and $n\equiv2,3\pmod{4}$ when $k$ is even.
\end{theorem}

\begin{theorem}\label{th:generalizedSS}
 Let $j=p^et$, where $p$ is the smallest prime factor of $j$, and $e,t$ are positive integers.
Then a $(j,n)$-generalized Skolem sequence exists if and only if $n\equiv0,1,\ldots,p-1\pmod{p^{e+1}}$.
\end{theorem}

There is an efficient heuristic algorithm for obtaining a
Skolem sequence of large order based on hill-climbing~\cite{constructionSS}. The sufficiency of
the above existence theorems is usually proved by giving direct constructions of the required sequences.
New sequences can also be found by concatenating two or more existing sequences.

Several direct constructions of hooked extended Skolem sequences
are provided in~\cite{linek} in terms of unions or sums of two Skolem sequences. Pivoting
and doubling are another techniques for constructing Skolem sequences~\cite{surveySS}.
When $n>5$, we can use the explicit constructions given in~\cite{handbook}
to find our ordered pairs for Skolem and hooked Skolem sequences.

By means of various types of Skolem sequences we introduce
Skolem interleavers with known cycle structure.
Also we can provide
interleavers with cycles of length $j$ or $1$ by using
generalized Skolem sequences.

\section{New algebraic interleavers over $\mathbb{F}_q$}~\label{NewInterleaversFq}

\subsection{Permutation functions as deterministic interleavers}~\label{GeneralStructureofNewDeterministicInterleavers}
As before, let $\alpha$ be a primitive element, root of a primitive polynomial $f$
over $\mathbb{F}_p$ of degree $m$, and $q=p^m$. If $P$ is a
permutation function over
$\mathbb{F}_q=\{0,\alpha^1,\ldots,\alpha^{q-2},\alpha^{q-1}\}$,
then for every $1\leq i\leq q-1$ there exists a $j$
such that $P(\alpha^i)=\alpha^j$, $1\leq j\leq q-1$.
Thus, $P^{-1}(\alpha^j)=\alpha^i$.
In this manner, we can take this permutation function $P$
as a function which rearranges the powers of $\alpha$. Therefore,
every permutation function over the finite field $\mathbb{F}_q$
can induce an interleaver as follows:

\begin{definition}\label{PPinteleaver}
Let $\alpha$ be a root of a primitive polynomial $f$ over
$\mathbb{F}_p$ of degree $m$, and $q=p^m$. Let $P$ be a
permutation function over $\mathbb{F}_q$. An interleaver
$\Pi_P:\mathbb{Z}_{q}\rightarrow\mathbb{Z}_{q}$ is defined by
\begin{equation}\label{PPint}
\Pi_P(i)=\ln (P(\alpha^i))
\end{equation}
where $\ln(.)$ denotes the discrete logarithm to the base $\alpha$
over $\mathbb{F}_q^\ast$ and $\ln(0)=0$.
\end{definition}

It is easy to see that every permutation function $P$ has a
unique compositional inverse $P^{-1}$. Clearly, $P^{-1}$ is also
a permutation function over $\mathbb{F}_q$.
We can also define
the interleaver $\Pi_{P^{-1}}$ by means of $P^{-1}$.
Based on the above discussions the following statments can be described~\cite{allertonversion}.
There is a one-to-one correspondence between the set of all permutation functions
over a fixed finite field $\mathbb{F}_q$ and the set of all
interleavers of size $q$.
One of the straightforward consequences of the above facts is that
for a self-inverse permutation function $P$ over $\mathbb{F}_q$, we have $\Pi_P=(\Pi_P)^{-1}$.

We now proceed to introduce interleavers based
on permutation functions over finite fields based on the above discussions.
The following general definition works for all of them.
\begin{definition}\label{monint}
Let $f$ be one of the permutation functions over $\mathbb{F}_q$
that were cited in Section~\ref{BasicDefinitionsandBackground}.
Then $\Pi_f$ as defined in~(\ref{PPint})
is a new determinstic interleaver. Each of them can be
explicitly named by their underlying permutation function.
\end{definition}
For example, we have {\em monomial, Dickson, M\"{o}bius, R\'edei and Skolem interleavers}.
The classes of permutation functions with explicit inverse formulas
are cited in Section~\ref{BasicDefinitionsandBackground}.
Each of them induce a new determinstic interleaver based on the above definition.

\subsection{Monomial interleavers}\label{MonInterleavers}

Let $M(x)=x^n$ over $\mathbb{F}_q$ and $\gcd(n,q-1)=1$. Then
$$\Pi_M(i)=\ln(M(\alpha^i))=\ln((\alpha^i)^n)=ni\!\!\!\!\pmod{q-1},$$
for $i\in\mathbb{Z}_{q}$. So, $\Pi_M(x)=nx\pmod{q-1}$ for
$x\in\mathbb{Z}_{q}$. Since $\gcd(n,q-1)=1$, $\Pi_M$ is a
linear permutation polynomial.
\begin{ex}\label{ex:monint}
Assume that $n=11$ and $q=13$. Since $\gcd(11,12)=1$, the
monomial $M(x)=x^{11}$ is a permutation polynomial
over $\mathbb{F}_{13}$.
Furthermore, $11.11\equiv1\pmod{12}$ and this means that $M^{-1}=M$
is a permutation polynomial over $\mathbb{F}_{13}$ and $\Pi_{M^{-1}}=(\Pi_{M})^{-1}$.
Therefore, $\Pi_{M^{-1}}=\Pi_M$ can act as the deinterleaver too.
Since $2$ is a primitive element of $\mathbb{F}_{13}$, we get
$$\begin{array}{llll}
M(2^1)=2^{11},&M(2^2)=2^{10},&M(2^3)=2^{9},&M(2^4)=2^{8},\\
M(2^7)=2^{5},&M(2^6)=2^{6},&M(2^5)=2^{7},&M(2^8)=2^4,\\
M(2^9)=2^{3},&M(2^{10})=2^{2},&M(2^{11})=2^{1},&M(2^{12})=2^{12}.
\end{array}$$
We can interpret this interleaver and its deinterleaver using
$$\left(\begin{array}{lllllllllllll}
0&1&2&3&4&5&6&7&8&9&10&11&12\\
0&11&10&9&8&7&6&5&4&3&2&1&12
\end{array}\right).$$
Also we observe that the only three fixed points are $0$,
$2^6=-1\pmod{13}$ and $2^{12}=1\pmod{13}$; this is a general
fact as it will be seen in Corollary~\ref{cor:monPPsamelength}.
\end{ex}
Our approach
leads us to use self-inverse permutation functions which
have cycles of the same length $j=2$, or otherwise fixed
points. Such permutation monomials are obtained using
the following theorem from~\cite{rob2} for $j=2$. We recall
that $j=\mbox{ord}_s(n)$, if $j$ is the smallest integer
with the property $n^j\equiv1\pmod{s}$.
\begin{theorem}\label{th:monPPsamelength}
Let $q-1=p_0^{k_0}p_1^{k_1}\ldots p_r^{k_r}$. The permutation
monomial $M(x)=x^n$ of $\mathbb{F}_q$ has only cycles of the
same length $j$ or $1$ (fixed points) if and only if one of
the following conditions holds for each $0\leq \ell\leq r$:
\begin{itemize}
\item $n\equiv1\pmod{p_\ell^{k_\ell}}$,
\item $j=\mbox{ord}_{p_\ell^{k_\ell}}(n)$ and $j|p_\ell-1$,
\item $j=\mbox{ord}_{p_\ell^{k_\ell}}(n)$, $k_\ell\geq2$ and $j=p_\ell$.
\end{itemize}
\end{theorem}
Since we concentrate on the case $j=2$, the following
corollary is useful for us.
\begin{corollary}\label{cor:monPPsamelength}
Let $q-1=p_0^{k_0}p_1^{k_1}\ldots p_r^{k_r}$ where $p_0=2$.
The permutation polynomial of $\mathbb{F}_q$ given by
$M(x)=x^n$ decomposes in cycles of the same length $j$ and
$\{0,1,-1\}$ are the only fixed elements if and only if
\begin{itemize}
\item{for $k_0>2$:} $j=2$ and $n=q-2$ or $\displaystyle n=\frac{q-3}{2}$.
\item{for $k_0=2$:} $j=2$ and $n=q-2$.
\end{itemize}
\end{corollary}
We note that if we use $n=q-2$, we get some sort of symmetry
in our permutation as in Example~\ref{ex:monint}. Let $n=q-2$,
then we have $\Pi_M(x)$ equals to
$$\ln{\left(M\left(\alpha^x\right)\right)}=
  \ln{\left(\alpha^{x(q-2)}\right)}=x(q-2)\!\!\!\!\pmod{q-1}.$$
It is easy to see that for every $i\in\mathbb{Z}_q$ we have
$\Pi_{M}(i)=q-1-i$ and $\Pi_{M}(q-1-i)=i$ because
$\Pi_M(i)=i(q-2)=i(q-1-1)=i(q-1)-i=-i=q-1-i\pmod{q-1}$
and since the permutation is self-inverse, $\Pi_M(q-1-i)=i$.
However, this is not the case for $\displaystyle n=\frac{q-3}{2}$
in Corollary~\ref{cor:monPPsamelength}. In this situation
$\displaystyle \Pi_M(x)=x\left(\frac{q-3}{2}\right)\pmod{q-1}$.
We do not see that symmetry in this case and it seems that
these self-inverse monomial interleavers perform better than
self-inverse monomial interleavers with $n=q-2$. The authors
of~\cite{cor} have constructed and investigated the efficiency of
self-inverse monomial interleavers with $n=q-2$.

\subsection{Dickson interleavers}\label{DicksonInterleavers}

Let $\gcd(n,q^2-1)=1$. It is known~\cite{lidl} that $D_n(x,a)$
for $a\in\{0,\pm1\}$ is a permutation polynomial over
$\mathbb{F}_q$ and has the compositional inverse $D_m(x,a)$
where $nm\equiv1\pmod{q^2-1}$. We can define a set of
deterministic interleavers
$\displaystyle \Pi_D^{(n,a)}:\mathbb{Z}_q\rightarrow\mathbb{Z}_q$
by $\displaystyle \Pi_D^{(n,a)}(i)=\ln(D_n(\alpha^i,a))$.
\begin{ex}\label{ex:Dicksonint}
Let $n=19$, $q=11$ and $a=1$. Then we get
$$D_{19}(x,1)=x^{9}+3x^{7}+9x^5+5x^3+5x\!\!\!\!\pmod{11}.$$
Since $\gcd(19,(11)^2-1)=\gcd(19,120)=1$ and $a=1$, $D_{19}(x,1)$
is a permutation polynomial over $\mathbb{F}_{11}$ with
compositional inverse $D_m(x,1)$ where $19m\equiv 1\pmod{120}$.
Therefore, $m=19$, and this means that $D_{19}(x,1)$ is a
self-inverse Dickson permutation polynomial over
$\mathbb{F}_{11}$. A Dickson interleaver
$\Pi_D^{(19,1)}:\mathbb{Z}_{11}\rightarrow\mathbb{Z}_{11}$
can be defined by $\Pi_D^{(19,1)}(i)=\ln(D_{19}(2^i,1))$
where $2\in\mathbb{F}_{11}$ is a primitive element.
Thus, we have the following:
$$\begin{array}{llll}
D_{19}(0,1)=0,&D_{19}(2,1)=2,&D_{19}(2^2,1)=2^2,\\
D_{19}(2^3,1)=2^3,&D_{19}(2^4,1)=2^9,&D_{19}(2^5,1)=2^5,\\
D_{19}(2^6,1)=2^6,&D_{19}(2^7,1)=2^7,&D_{19}(2^8,1)=2^8,\\
D_{19}(2^9,1)=2^4,&D_{19}(2^{10},1)=2^{10}.&
\end{array}$$
This means that $\Pi_D^{(19,1)}$ permutes the elements
of $\mathbb{Z}_{11}$ as follow
$$\left(\begin{array}{ccccccccccc}
0&1&2&3&4&5&6&7&8&9&10\\
0&1&2&3&9&5&6&7&8&4&10
\end{array}\right).$$
\end{ex}
The following two theorems are from~\cite{rob1}. We use
$j=\mbox{ord}^-_s(n)$ for the least integer with
$n^j\equiv-1\pmod{s}$.
\begin{theorem}\label{th:DickPPsamelength1}
Let $q-1=p_1^{k_1}\ldots p_r^{k_r}$ and
$q+1=p_{r+1}^{k_{r+1}}\ldots p_s^{k_s}$ be the
prime factorization of $q-1$ and $q+1$, respectively.
Suppose that $\gcd(n,q^2-1)=1$.
The Dickson permutation polynomial $D_n(x,1)$ over
$\mathbb{F}_q$ is the identity on $\mathbb{F}_q$
or all the non-trivial cycles have length two if and only if one of the following holds for all
$1\leq \ell\leq r$, and one of the following conditions
holds for all $r+1\leq \ell\leq s$:
\begin{enumerate}
\item{Either}
\begin{enumerate}
\item $n\equiv1\pmod{p_\ell^{k_\ell}}$ and $p_\ell^{k_\ell}=2$, or
\item $2=\mbox{ord}_{p_\ell^{k_\ell}}^{-}(n)$, and $4|(p_\ell-1)$.
\end{enumerate}
\item{Either}
\begin{enumerate}
\item $n\equiv\pm1\pmod{p_\ell^{k_\ell}}$, or
\item $2=\mbox{ord}_{p_\ell^{k_\ell}}(n)$, $p_\ell=2$, $k_\ell\geq2$, and $n\not\equiv -1\pmod{p_\ell^{k_\ell}}$.
\end{enumerate}
\end{enumerate}
\end{theorem}
\begin{theorem}\label{th:DickPPsamelength-1}
Let $q-1=p_1^{k_1}\ldots p_r^{k_r}$
and $q+1=p_{r+1}^{k_{r+1}}\ldots p_s^{k_s}$
be the prime factorization of $q-1$ and $q+1$, respectively.
Suppose that $\gcd(n,q^2-1)=1$.
The Dickson permutation polynomial $D_n(x,-1)$ over
$\mathbb{F}_q$ is the identity on $\mathbb{F}_q$
or all the non-trivial cycles have length two if and only if
one of the following conditions holds for all $1\leq \ell\leq r$,
and one of the following conditions holds for all $r+1\leq \ell\leq s$
\begin{enumerate}
\item{Either}
\begin{enumerate}
\item $2(n+1)\equiv0\pmod{p_\ell^{k_\ell}}$ and $p_\ell^{k_\ell}=2,4$, or
\item $2=\mbox{ord}_{p_\ell^{k_\ell}}^{-}(n)$, and $4|(p_\ell-1)$
\end{enumerate}
\item{Either}
\begin{enumerate}
\item $2(n+1)\equiv0\pmod{p_\ell^{k_\ell}}$, or
\item $n\equiv1\pmod{p_\ell^{k_\ell}}$, or
\item $2=\mbox{ord}_{p_\ell^{k_\ell}}(n)$, $k_\ell\geq2$ and $p_\ell=2$.
\end{enumerate}
\end{enumerate}
\end{theorem}
If $D_n(x,1)$ meets the assumptions of
Theorem~\ref{th:DickPPsamelength1}, then all the cycles of
$D_n(x,1)$ have the same length $1$ or $2$. This means that
$D_n(x,1)$ is a self-inverse permutation polynomial, that is,
$D_n(x,1)=D_m(x,1)$ where $mn\equiv1\pmod{q^2-1}$.
Hence, we conclude that $\Pi_{D}^{(n,1)}$
is a self-inverse interleaver. A similar statement can be given
for $D_n(x,-1)$ using Theorem~\ref{th:DickPPsamelength-1}.

Furthermore, by using the above theorems and other results
in~\cite{rob1} we can derive Dickson polynomials which produce
permutations with prescribed cycle structure. Hence these
polynomials can be employed to construct both self-inverse
and non-self-inverse interleavers. If $a=0$ then these Dickson
polynomials turn out to be monomials. Thus, their corresponding
Dickson interleavers can be considered as a generalization of
monomial interleavers.

\subsection{M\"obius interleavers}\label{MobiusInterleavers}

In the following we establish and introduce M\"{o}bius
interleavers. This is the first usage of these non-linear
transformations to construct interleavers.
In order to generate M\"{o}bius interleavers with prescribed
cycle arrangement, one should know the cycle structure of
these functions. This is reported from~\cite{turkey} and
used here to provide new deterministic interleavers. Since
self-inverse M\"{o}bius interleavers need only three defining
parameters $a=d$, $b$ and $c$, they have a simple and efficient
structure that is easy to implement.
One can derive the inverse function of $T$ using~(\ref{MobinversePP}).
We have $T=T^{-1}$ if and only if $a=d$, $-b=b$
and $c=-c$. Let $q=2^n$, we get $-b=b$
and $c=-c$. Therefore, for a self-inverse M\"{o}bius function we have
\begin{equation}\label{eq:invMobiuschar2}
T(x)=T^{-1}(x)=\left\{
\begin{array}{ll}
 \frac{ax+b}{cx+a}& x\neq \frac{a}{c},\\
\frac{a}{c}& x=\frac{a}{c},
\end{array}
\right.
\end{equation}
where $a^2-bc\neq0$ and $c\neq 0$.
A detail example is provided in~\cite{allertonversion}.

We cite next theorem from~\cite{turkey}. This fully describes the cycle structure of
$T$ in terms of the eigenvalues of the coefficient matrix $A_T$
associated to $T$
\begin{equation}\label{eq:matrixcoeffMobius}
A_T=\left(\begin{array}{cc}
a&b\\c&d
\end{array}\right).
\end{equation}
The characteristic polynomial $t$ of $A_T$ is a
quadratic polynomial.
\begin{theorem}\label{th:cyclestructureMobius}
Let $T$ be the permutation defined by~(\ref{MobPP}), and let
$t$ be the characteristic polynomial of the matrix $A_T$
associated to $T$. Let $\alpha_1,\alpha_2\in\mathbb{F}_{q^2}$
be the roots of $t$.
\begin{enumerate}
\item{} Suppose $t(x)$ is irreducible. If
$\displaystyle k=\mbox{ord}\left(\frac{\alpha_1}{\alpha_2}\right)
=\frac{q+1}{s}$, $1\leq s<\frac{q+1}{2}$, then $T$ has $s-1$
cycles of length $k$ and one cycle of length $k-1$. In particular
$T$ is a full cycle if $s=1$.
\item{} Suppose $t(x)$ is reducible and
$\alpha_1,\alpha_2\in\mathbb{F}_q^\ast$ are roots of $t(x)$
and $\alpha_1 \neq \alpha_2$. If
$\displaystyle k=\mbox{ord}\left(\frac{\alpha_1}{\alpha_2}\right)
=\frac{q-1}{s}$, $s\geq1$, then $T$ has $s-1$ cycles of length
$k$, one cycle of length $k-1$ and two cycles of length $1$.
\item{}Suppose $t(x)=(x-\alpha_1)^2$, $\alpha_1\in\mathbb{F}_q^\ast$
where $q=p^n$. Then $T$ has $p^{n-1}-1$ cycles of length $p$,
one cycle of length $p-1$ and one cycle of length $1$.
\end{enumerate}
\end{theorem}
It is obvious that we are interested again on permutations
with cycles of length $1$ and $2$ only. Hence based on the above theorem we
have the following theorem:
\begin{theorem}\label{th:cyclestructureMobiusnew}
Let $\Pi_T$ be an interleaver defined by $T$, and let
$t$, $A_T$, $\alpha_1$ and $\alpha_2$ be as in Theorem~\ref{th:cyclestructureMobius}.
Then $\Pi_T$ is a self-inverse interleaver if $\tr(A_T)=0$.
\end{theorem}

\begin{proof}
The three cases of the previous theorem have the following consequences:
\begin{enumerate}
\item{}We should have
$\displaystyle 2=k=\mbox{ord}\left(\frac{\alpha_1}{\alpha_2}\right)$.
But $k=2$ if and only if $(\alpha_1)^2=(\alpha_2)^2$ and $\alpha_1\neq\alpha_2$
if and only if $\alpha_1=-\alpha_2$ if and only if $a+d=\tr(A_T)=\alpha_1+\alpha_2=0$.
Hence, $t$ is irreducible and $\tr(A_T)=0$ if and only if $k=2$ and we have
$\frac{q+1}{2}-1$ cycles of length two and one cycle of length one.
\item{}In a similar situation with case 1) we have that $t$ is
reducible and $\tr(A_T)=0$ if and only if $k=2$ and we have
$\displaystyle \frac{q-1}{2}-1$
cycles of length two and three cycles of length one.
\item{}In the third case we have only cycles of length $1$ and $2$
when $p=2$. As we mentioned in~(\ref{eq:invMobiuschar2}) in this
case we also have $\tr(A_T)=a+d=0$. So $a=d$ and this means that
$\alpha_1=a$. Thus, $a=d$ if and only if $T$ has $2^{n-1}-1$ cycles of length
$2$ and two cycles of length $1$ where $q=2^n$.
\end{enumerate}
\end{proof}

\subsection{R\'edei interleavers}\label{RedeiInterleavers}

In this section the concept of R\'edei interleavers is introduced.
Again, obtaining the cycle structure of R\'edei functions is
essential. In the following, the arrangement of cycles of R\'edei
functions is described. More precisely, R\'edei functions with
all cycles of length $j\neq1$ are derived. We also provide a
condition under which all cycles of a R\'edei function are of
length $j$ or $1$.
The inverse for every R\'edei function $R_n$ is provided next.
The following relations are cited from~\cite{carlitz} for all $x,y\in\mathbb{F}_q$
\begin{eqnarray}
   &&R_n(R_m)(x,a)=R_{nm}(x,a),\label{eq:compositRedei}\\
   &&R_n(x,a)=x \Longleftrightarrow n
       \equiv1\!\!\!\!\pmod{q+1},\label{eq:idRedei}\\
   &&R_n\left(\frac{xy+a}{x+y},a\right)=
       \frac{R_n(x,a)R_n(y,a)+a}{R_n(x,a)+R_n(y,a)}\label{eq:propertyRedei}.
\end{eqnarray}
The next lemma is proved in~\cite{carlitz}.
\begin{lemma}\label{th:generalproperty}
Let $r$ be a rational function with coefficients in $\mathbb{F}_q$
that satisfies
\begin{equation}\label{eq:generalproperty}
 r\left(\frac{xy+a}{x+y}\right)=\frac{r(x)r(y)+a}{r(x)+r(y)}
\end{equation}
where $a$ is a fixed element of $\mathbb{F}_q$ and $x$ and
$y$ are two unknowns. Then, if $a\neq 0$ and
$\mbox{char}(\mathbb{F}_q)\neq 2$, $r$ coincides with a
R\'edei's function for some $m$ (not necessarily relatively
prime to $q+1$).
\end{lemma}
Comparing~(\ref{eq:generalproperty}) with~(\ref{eq:propertyRedei})
results in the next theorem which provides the inverse of every R\'edei function.
\begin{theorem}\label{th:Redeiinverse}
Let $R_n$ for some $n$ be a R\'edei function over $\mathbb{F}_q$
where $a\in\mathbb{F}_q^\ast$ is a non-square and $\gcd(n,q+1)=1$.
Then $R_n^{-1}=R_m$ for $m$ satisfying $nm\equiv1\pmod{q+1}$.
\end{theorem}

\begin{proof}
First, it is clear that $R_n$ has a compositional inverse
$R_n^{-1}$. Using~(\ref{eq:propertyRedei}) and
taking $R_n^{-1}(.)$ on both sides, we get
$$ \frac{xy+a}{x+y}
 = R_n^{-1}\left(\frac{R_n(x,a)R_n(y,a)+a}{R_n(x,a)+R_n(y,a)}\right).$$
Let us assume that $s=R_n(x,a)$ and $t=R_n(y,a)$. Then
applying~(\ref{eq:generalproperty}) to $R^{-1}_n$, we get
$$ \frac{R_n^{-1}(R_n(x,a))R_n^{-1}(R_n(y,a))+a}
        {R_n^{-1}(R_n(x,a))+R_n^{-1}(R_n(y,a))}  
  = R_n^{-1}\left(\frac{R_n(x,a)R_n(y,a)+a}{R_n(x,a)+R_n(y,a)}\right),
$$
implying that, for all $s,t\in\mathbb{F}_q$, we get
$$ \frac{R_n^{-1}(s)R_n^{-1}(t)+a}{R_n^{-1}(s)+R_n^{-1}(t)}
 = R_n^{-1}\left(\frac{st+a}{s+t}\right).$$
Since $R_n^{-1}$ satisfies all the conditions of
Lemma~\ref{th:generalproperty}, $R_n^{-1}$ coincides with a
R\'edei function for some $m$. So, $R_n^{-1}=R_m$. Now, we
use~(\ref{eq:compositRedei}) and~(\ref{eq:idRedei}) to get
$$  id_{\mathbb{F}_q}=R_n(R_n^{-1})=R_n(R_m)=R_{nm} 
    \Longleftrightarrow nm\equiv1\!\!\!\!\!\pmod{q+1}. $$
\end{proof}

We note that $R_n^{-1}$ is a rational function because every
function from $\mathbb{F}_q$ to itself can be interpreted as
a polynomial with degree less than $q$.

An example is given in~\cite{allertonversion}.
\begin{theorem}\label{th:acyclejRedei}
Let $j$ be a positive integer. The R\'edei function
$R_n(x,a)$ of $\mathbb{F}_q$ with $\gcd(n,q+1)=1$
has a cycle of length $j$ if and only if
$q+1$ has a divisor $s$ such that $j=\mbox{ord}_s(n)$.
\end{theorem}

\begin{proof}
Let $R_n^{(j)}(x,a)$ denote the $j$-th iterate of
$R_n(x,a)$ under the composition operation.
We get
\begin{eqnarray}
R_n^{(j)}(x,a)=x&\Longleftrightarrow& R_{n^{j}}(x,a)
    =x\Longleftrightarrow\frac{G_{n^j}(x,a)}{H_{n^j}(x,a)}=x\nonumber\\
&\Longleftrightarrow& G_{n^j}(x,a)=xH_{n^j}(x,a),\nonumber
\end{eqnarray}
where
the first equivalence is derived from~(\ref{eq:compositRedei}).
Furthermore, we get the following identities for $x\in\mathbb{F}_q$
\begin{eqnarray}
   ({x+\sqrt{a}})^{n^{j}}&=&G_{n^{j}}(x,a)+H_{n^{j}}(x,a)\sqrt{a}\nonumber\\
   &=&xH_{n^{j}}(x,a)+H_{n^{j}}(x,a)\sqrt{a}\nonumber\\
   &=&H_{n^{j}}(x,a)(x+\sqrt{a}).\nonumber
\end{eqnarray}
Let us assume that $y={x+\sqrt{a}}$, then
$y^{n^j-1}=H_{n^j}(x,a)\in\mathbb{F}_q$
and $y\in\mathbb{F}_{q^2}$.
Since $y^{n^j-1}\in\mathbb{F}_q$ we have that
$y^{(n^j-1)(q-1)}=1$. So, $R_n(x)$ has a cycle of length
$j$ if and only if $y^{(n^j-1)(q-1)}=1$ if and only if
$q^2-1$, which is the size of the multiplicative group
$\mathbb{F}^\ast_{q^2}$, has a divisor $t$ such that
$t|(n^{j}-1)(q-1)$ if and only if $q+1$ has a divisor $s$
where $n^j\equiv1\pmod{s}$ and $j$ is the smallest integer
with this property.
\end{proof}

Based on the above theorem and a simple counting technique,
we can find the explicit number of cycles of length $j$ for a
R\'edei function $R_n$ over $\mathbb{F}_q$.
\begin{theorem}\label{th:numberoflengthjRedei}
The number $N_j$ of cycles of length $j$ of the R\'edei function
$R_n$ over $\mathbb{F}_q$ with $\gcd(n,q+1)=1$ satisfies
$$jN_j+\sum_{\substack{i|j\\ i<j}} iN_i+1=\gcd(n^{j}-1,q+1).$$
\end{theorem}

\begin{proof}
Similar to the proof of Theorem~\ref{th:acyclejRedei}, we look
for $y\in\mathbb{F}^\ast_{q^2}$ such that\linebreak $y^{(n^j-1)(q-1)}=1$
and $y^{(n^j-1)}\in\mathbb{F}_q$. Let $\rho$ be a primitive
element of $\mathbb{F}^\ast_{q^2}$.
Let us assume that $s_0$ is a common divisor of $q+1$ and $n^j-1$.
Every $c$ with the property $\gcd(c,q+1)=\frac{q+1}{s_0}$ can
raise to a cycle of length $j$
$$\left(y,R_n(y),R_n^{(2)}(y),\ldots,R_n^{(j-1)}(y)\right)=
  \left(\rho^c,R_n(\rho^c),R_n^{(2)}(\rho^c),\ldots,R_n^{(j-1)}(\rho^c)\right).$$
On the other hand $\gcd(c,q+1)=\frac{q+1}{s_0}$ implies that
there exists $t_0\in\mathbb{N}$ such that $c=\frac{q+1}{s_0}t_0$.
Hence, we get
$$ \left(\rho^c\right)^{(n^{j}-1)(q-1)}=\left(\rho\right)^{c(n^{j}-1)(q-1)}=
   \left(\rho\right)^{\frac{n^{j}-1}{s_0}(q^2-1)t_0} 
 = \left(1\right)^{\frac{n^{j}-1}{s_0}t_0}=1, 
$$
and
$$ \left(\rho^c\right)^{(n^{j}-1)}=\left(\rho\right)^{c(n^{j}-1)}
 = \left(\rho\right)^{\frac{n^{j}-1}{s_0}(q+1)t_0}
 = \left(\rho^{q+1}\right)^{\frac{n^{j}-1}{s_0}t_0}\in\mathbb{F}_q,$$
where the last expression is true because the powers of $q+1$
of $\rho$ form $\mathbb{F}_q$ in $\mathbb{F}_{q^2}$. We are
interested in the number of $c=\frac{q+1}{s_0}t_0$ such that
$\gcd(c,q+1)=\frac{q+1}{s_0}$. This is the number of $t_0$
such that $\gcd(t_0,s_0)=1$ and $t_0\leq s_0$, which equals
$\phi(s_0)$ where $\phi(\cdot)$ denotes Euler's function.
Therefore, summing over all $\phi(s_0)$ with $s_0|\gcd(q+1,n^j-1)$
gives the number of elements that contribute to the cycles of
length $j$ and all its divisors $i$. We have
$$1+\sum_{i|j}iN_i=\sum_{s_0|\gcd(q+1,n^j-1)}\phi(s_0)=\gcd(q+1,n^j-1).$$
The last equality is derived from the fact that for every $n$
we have $\sum_{d|n}\phi(d)=n$. We note that 1 in the left hand
side accounts for the element $0$ since $R_n(0)=0$ gives an extra
fixed point and the above counting enumerates non-zero elements
$y\in\mathbb{F}_{q^2}$.
\end{proof}

Now, we are in a position to state our final expression. The following
theorem gives a general condition for a R\'edei function to have only cycles of length $j$ and $1$.
\begin{theorem}\label{th:allsamelengthRedei}
The R\'edei function $R_n$ of $\mathbb{F}_q$ with $\gcd(n,q+1)=1$
has all its cycles of length $j$ or $1$ if and only if for every
divisor $s$ of $q+1$ we have $n\equiv 1\pmod{s}$ or $j=ord_{s}(n)$.
\end{theorem}

\begin{corollary}\label{cor:allsamelength2Redei}
The R\'edei function $R_n$ of $\mathbb{F}_q$ with $\gcd(n,q+1)=1$
is self-inverse if and only if $n^2\equiv1\pmod{q+1}$.
\end{corollary}

\begin{proof}
One can insert $j=2$ in Theorem~\ref{th:allsamelengthRedei} and produce a self-inverse
R\'edei function. Thus, $R_n(x)$ is a self-inverse
function if and only if
\begin{equation}\label{eq:selfinverseRedei}
 \mbox{for~all}~s|q+1,\quad\left\{
 \begin{array}{l}
  n\equiv 1\pmod{s}~\mbox{or}\\
  n^2\equiv 1\pmod{s}.
 \end{array}\right.
\end{equation}
It is easy to see that~(\ref{eq:selfinverseRedei})
is equivalent to $n^2\equiv1\pmod{q+1}$ and this simply
shows that $m=n$ for $R_m(x)=R^{-1}_n(x)$.
\end{proof}

Now we are able to find all self-inverse R\'edei functions
by using the above corollary. Also, one can apply self-inverse
R\'edei functions to produce self-inverse R\'edei interleavers.
\begin{theorem}\label{th:factors}
 Let $q+1=p_0^{k_0}p_1^{k_1}\ldots p_r^{k_r}$. The permutation of $\mathbb{F}_q$
given by the R\'edei function $R_n$ has cycles of the same length $j$
or fixed points if and only if one of the following conditions holds for each
$0\leq \ell\leq r$
\begin{itemize}
\item $n\equiv1\pmod{p_\ell^{k_\ell}}$,
\item $j=\mbox{ord}_{p_\ell^{k_\ell}}(n)$ and $j|p_\ell-1$,
\item $j=\mbox{ord}_{p_\ell^{k_\ell}}(n)$, $k_\ell\geq2$ and $j=p_\ell$.
\end{itemize}
\end{theorem}

\begin{proof}
We begin stating some easy lemmas and propositions.
\begin{lemma}\label{lem:1}
If $n\equiv b\pmod{p^\ell}$, then $n^p\equiv b^p\pmod{p^{\ell+1}}$
for all $l\geq1$.
\end{lemma}
\begin{lemma}\label{lem:2}
Let $j=\mbox{ord}_{p^\ell}(n)$. Then $j=\mbox{ord}_{p^{\ell+1}}(n)$
or $jp=\mbox{ord}_{p^{\ell+1}}(n)$.
\end{lemma}
\begin{proposition}\label{prop:1}
We have that $j=\mbox{ord}_{p^k}(n)$ and $j|p-1$ if and only if $j=\mbox{ord}_{p^\ell}(n)$ for all $1\leq \ell\leq k$.
\end{proposition}
\begin{lemma}\label{lem:3}
Let $p=\mbox{ord}_{p^k}(n)$ for some $k\geq2$. Then either
$2=p=\mbox{ord}_{p^\ell}(n)$ for $2\leq \ell\leq k$ or
$n\equiv1\pmod{p^\ell}$ for $1\leq \ell\leq k$.
\end{lemma}
\begin{lemma}\label{lem:4}
Let $j=\mbox{ord}_{s}(n)$, $j=\mbox{ord}_{\ell}(n)$ and $\gcd(s,\ell)=1$.
Then $j=\mbox{ord}_{s\ell}(n)$.
\end{lemma}
\begin{lemma}\label{lem:5}
Let $j=\mbox{ord}_{s}(n)$, $n\equiv 1\pmod{\ell}$ and $\gcd(s,\ell)=1$.
Then $j=\mbox{ord}_{s\ell}(n)$.
\end{lemma}
Now we proceed to give the main proof for Theorem~\ref{th:factors}
($\Longleftarrow$) If $n\equiv1\pmod{p_\ell^{k_\ell}}$
for all $0\leq \ell\leq r$, then $R_n(x)$ is the identity
permutation. Suppose that $1<j=\mbox{ord}_{p_\ell^{k_\ell}}(n)$
for some of the $l$'s and $n\equiv1\pmod{p_\ell^{k_\ell}}$
for the others. Proposition~\ref{prop:1} and Lemma~\ref{lem:3}
guarantee that $j=\mbox{ord}_{p_\ell^k}(n)$ or $n\equiv1\pmod{p^k_\ell}$
for all $0\leq \ell\leq r$ and $1\leq k\leq k_\ell$. Now, if
$t|(q+1)$, then by Lemmas~\ref{lem:4} and~\ref{lem:5},
we have that, $j=\mbox{ord}_t(n)$ or $n\equiv1\pmod{t}$.
Hence, by Theorem~\ref{th:allsamelengthRedei},
all the cycles have length $j$ or $1$.\\
($\Longrightarrow$) Suppose that all the cycles
have the same length $j$. Then, by Theorem~\ref{th:allsamelengthRedei},
$j=\mbox{ord}_t(n)$ or $n\equiv1\pmod{t}$ for all $t$
that divides $q+1$. This holds in particular for $t=p_\ell^{k_\ell}$;
$0\leq \ell\leq r$. We only have to prove that, if
$j=\mbox{ord}_{p_\ell^{k_\ell}}(n)$ then $j|(p_\ell-1)$ or $j=p_\ell$; $k_\ell\geq2$.
Suppose that $1\neq j=\mbox{ord}_{p_\ell^{k_\ell}}(n)$.
If $k_\ell=1$ then $j|(p_\ell-1)$ and we are done.
If $k_\ell\geq2$ and $j\not|(p_\ell-1)$,
then Proposition~\ref{prop:1} implies that
$j\neq \mbox{ord}_{p_\ell^k}(n)$ for some $k<k_\ell$.
Let $s$ be the largest one such that $j\neq\mbox{ord}_{p_\ell^s}(n)$.
Then $n\equiv1\pmod{p_\ell^s}$ because otherwise, by
Theorem~\ref{th:allsamelengthRedei}, there would
be a cycle of length different from $j$. By Lemma~\ref{lem:1},
$i^{p_\ell}\equiv1\pmod{p_\ell^{s+1}}$. But $j=\mbox{ord}_{p^{s+1}_\ell}(n)$
implies that $j|p_\ell$ and hence $j=p_\ell$.
It has to be noted that the above proof is similar to the
proof of Theorem 2 in~\cite{rob2}.
\end{proof}

\subsection{Skolem interleavers}\label{SkolemInterleavers}

In this section we construct other types of interleavers. In
this case our underlying structure are various types of Skolem
sequences including $k$-extended, hooked and $(j,n)$-generalized.
Clearly, we shall use $(j,n)$-generalized Skolem sequences
to produce interleavers with cycles of length $1$ and $j$ only.
First of all, let us make a slight modification to every type
of Skolem sequences to turn them into a consistent form for
using in our construction method. Let us assume that
$S=(s_1,\ldots,s_t)$ be a Skolem sequence (not a generalized
one but maybe $k$-extended or even hooked) over a set of integers
$D$. For every $i\in [n]$ if $\ell$, $2\leq \ell\leq t$, is the
largest index such that $s_{\ell}=i$ then convert $s_{\ell}$ to
$-s_{\ell}$ i.e. put $-i$ instead of $i$ in the $\ell$th position.
In the case that $S$ is a $(j,n)$-generalized Skolem sequence
just insert $-(j-1)i$ instead of $i$ in the $\ell$th position.
These changed Skolem sequences are called \emph{modified Skolem
sequences} ({\em modified generalized Skolem sequences},
respectively). Now, our strategy is to reorder any set of integers
of length $t$, say $I=\{1,\ldots,t\}$, based on a modified
(generalized) Skolem sequence of length $t$, $S=(s_1,\ldots,s_t)$.

Let us assume that $S^m$ is a modified (generalized) Skolem sequence
of order $n$ with alphabet $[n]$. If $i\in [n]$ repeats on positions
$u$ and $v$ where $u<v$, then $s^m_u=i$ and $s^m_v=-i$ and $v-u=i$.
Now, define the interleaver $\Pi_S$ by sending $u$ to $v$. More
precisely
\begin{equation}\label{def:SSinterleavers}
\Pi_S(u)=u+s^m_u=u+i=v.
\end{equation}
We observe that holes or zeros of our sequence produce fixed points.
In other words, if $s^m_h=0$ for some $1\leq h\leq m$, then
$\Pi_S(h)=h+s^m_h=h+0=h$. If we imagine the indices of our modified
Skolem sequence as time and the amounts of our Skolem sequence as
location (domain), then this interleaver may be interpreted as a
combination of time and domain.
\begin{ex}
The sequence $S=(2,5,2,6,1,1,5,3,4,6,3,0,4)$ is a hooked Skolem
sequence. First we have to convert it to a modified hooked Skolem
sequence. If we denote the modified version of $S$ by $S^m$, we get
$$S^m=(2,5,-2,6,1,-1,-5,3,4,-6,-3,0,-4).$$
Thus, based on~(\ref{def:SSinterleavers}), we have
$$
\begin{array}{lll}
\Pi_S(1)=1+s^m_1=3,&
\Pi_S(2)=2+s^m_2=7,&
\Pi_S(3)=3+s^m_3=1,\\
\Pi_S(4)=4+s^m_4=10,&
\Pi_S(5)=5+s^m_5=6,&
\Pi_S(6)=6+s^m_6=5,
\end{array}
$$
$$
\begin{array}{lll}
\Pi_S(7)=7+s^m_7=2,&
\Pi_S(8)=8+s^m_8=11,&
\Pi_S(9)=9+s^m_9=13,\\
\Pi_S(10)=10+s^m_{10}=4,&
\Pi_S(11)=11+s^m_{11}=8,&
\Pi_S(12)=12+s^m_{12}=12,\\
\Pi_S(13)=13+s^m_{13}=9.&&
\end{array}
$$
The above equalities induce the following Skolem interleaver
$$\left(\begin{array}{ccccccccccccc}
1&2&3&4&5&6&7&8&9&10&11&12&13\\
3&7&1&10&6&5&2&11&13&4&8&12&9
\end{array}\right).$$
\end{ex}

\begin{theorem}\label{th:skolemint}
Let $\Pi_S$ be an interleaver constructed using a modified
Skolem sequence. Then $\Pi_S$ is a self-inverse interleaver.
Furthermore, if we use a modified $(j,n)$-generalized Skolem
sequence, then $\Pi_S$ has only cycles of length $j$ or $1$.
\end{theorem}

\begin{proof}
Let $S=(s_1,\ldots,s_m)$ and its modified version, $S^m$,
are used to produce $\Pi_S$. For every $i\in D$ there exist
indices $u$ and $v$ such that $u<v$, $s^m_u=i$, $s^m_v=-i$
and $v-u=i$. Based on the definition of $\Pi_S$ we get
$$\Pi_S(u)=u+s^m_u=u+i=v,\quad\Pi_S(v)=v+s^m_v=v-i=u,$$
and also if $s^m_h=0$ for some $h$ we have $\Pi_S(h)=h+s^m_h=h+0=h$.
Hence, $\Pi_S$ is a self-inverse interleaver.
On the other hand let us suppose that $S$ is a $(j,n)$-generalized Skolem sequence
and $S^m$ be its modified version. For every $i\in [n]$ there are indices
$r_1,r_2=r_1+i,\ldots,r_j=r_1+(j-1)i$ such that
$s^m_{r_1}=s^m_{r_2}=\cdots=s^m_{r_{j-1}}=i$ and $s^m_{r_j}=-(j-1)i$.
Since $r_{w+1}=r_w+i$ and based on~(\ref{def:SSinterleavers}) we get
$$\Pi_S(r_w)=r_w+s^m_{r_w}=r_w+i=r_{w+1}$$
for $1\leq w\leq j-1$, and for $w=j$ we have
$$\Pi_S(r_j)=r_j+s^m_{r_j}=r_j-(j-1)i=r_1.$$
It means that $(r_1,r_2,\ldots,r_j)$ builds a cycle of length $j$
for $\Pi_S$. More explicitly, $\Pi_S(r_w)=r_{w+1}$ for $1\leq w\leq j-1$
and $\Pi_S(r_j)=r_1$.
What we did for holes (zeros) of a modified Skolem sequence,
can be likewise done for generalized Skolem sequences.
\end{proof}

\section{Conclusion and further research}\label{conclusion}


Some deterministic interleavers have been introduced and investigated
based on permutation functions over finite fields. Well-known
permutation functions have been explained. R\'edei functions
treated in detail. We derived an exact formula for the inverse
of every R\'edei function. We gave the cycle structure of these functions
and provided the exact number of cycles of a certain length $j$.
Specifically, we focused on self-inverse permutation
functions which can produce self-inverse interleavers and have
the potential to reduce the memory consumption~\cite{ryu}.

The Skolem interleavers could be useful for making interleavers with a
known cycle structure. We have the same possible option and
characteristics for Dickson and R\'edei permutations as well.
More specifically, suppose that we are given a cycle structure
$(i_1,\ldots,i_n)$ for a permutation $\sigma$ where $i_j$
denotes the number of cycles of length $j$ in $\sigma$.
Clearly, we have $\sum_{j=1}^nji_j=n$. Assume that we want
to construct an interleaver following the structure of the
permutation $\sigma$.
We need $i_j$ cycles of length $j$. For example we can use a
$(j,ji_j)$-generalized Skolem sequence. Even monomial, Dickson
or R\'edei permutations can be employed to produce these cycles.
To this end, the cycle structure of monomial, Dickson and
R\'edei permutations can be determined using results of
Theorems~\ref{th:monPPsamelength},~\ref{th:DickPPsamelength1}
and~\ref{th:allsamelengthRedei}.

Clearly, more theoretical studies to provide maximum achievable
minimum distance, spread factor and dispersion~\cite{cor,cor2,rec}
of these interleavers are of interest.
Moreover, finding a suitable relationship
between cycles and sequences introduced in~\cite{dimitry}
and cycle structure of permutations seems worthwhile.

\medskip
Received November 2011; revised May 2012.

\medskip

{\it E-mail address:} amin$\_$sakzad@aut.ac.ir\\
\indent{\it E-mail address:} msadeghi@aut.ac.ir\\
\indent{\it E-mail address:} daniel@math.carleton.ca

\end{document}